\documentclass{article}
\usepackage{color,ulem,soul}
\usepackage{amsmath, amssymb, amsthm, braket, latexsym, mathtools}
\usepackage{subfiles,subcaption}
\usepackage{bm,array,arydshln}
\usepackage{graphicx}

\newtheorem{theorem}{Theorem}

\newtheorem{definition}{Definition}

\newcommand{\Z}{\mathbb{Z}}

\newcommand{\Add}[1]{\textcolor{black}{#1}}	
\newcommand{\Erase}[1]{\if0{#1}\fi}

\begin{document}
\newpage
\clearpage
\title{{\bf Parrondo's game of quantum search \\ based on quantum walk}
\vspace{15mm}} 

\author{
  Taisuke HOSAKA$^{\ast}$ \\
  College of Engineering Science \\
  Yokohama National University \\
  Hodogaya, Yokohama, 240-8501, Japan \\
  e-mail: hosaka-taisuke-pn@ynu.jp \\
  \\ \\
  Norio KONNO \\
  Department of Applied Mathematics, Faculty of Engineering \\ 
  Yokohama National University \\
  Hodogaya, Yokohama, 240-8501, Japan \\
  e-mail: konno-norio-bt@ynu.ac.jp 
}

\date{\empty }

\maketitle

\vspace{80mm}
\noindent
\begin{small}
{\bf Corresponding author$^{\ast}$}: Taisuke Hosaka, 
College of Engineering Science, 
Yokohama National University, 
Hodogaya, Yokohama, 240-8501, 
JAPAN, e-mail: hosaka-taisuke-pn@ynu.jp,
Tel.: +81-45-339-4205, Fax: +81-45-339-4205
\end{small}


\vspace{20mm}

\par\noindent

\clearpage

\begin{abstract}
  The Parrondo game, devised by Parrondo, means that \Erase{two losing games 
  can produce a winning game by a suitable combination of two losing games.}
  \Add{winning strategy is constructed a combination of losing strategy.}
  This situation is called the Parrondo paradox.
  The Parrondo game based on quantum walk and the search algorithm via quantum walk have been widely studied, respectively.
  This paper newly presents a Parrondo game of quantum search based on quantum walk by combining both models.
  Moreover we confirm that Parrondo's paradox exists for our model \Add{on the one- and two-dimensional torus} by numerical simulations.
  Afterwards we show the range in which the paradox occurs is symmetric about the origin on \Add{the $d$-dimensional torus $(d \geq 1)$} with even vertices and one marked vertex.
\end{abstract}
  
\vspace{10mm}
  
\begin{small}
\end{small}
\vspace{10mm}

\section{\bf Introduction \label{sec1}}
Quantum walk (QW), motivated from classical random walk (RW), 
has been studied since around 2000.
QW has different features as compared to RW.
One of the features is localization, that is, the probability of finding quantum walker is positive in the long time.
Because of its properties, QW plays an important role in the quantum search algorithm, see \cite{A,AKR,LS,P,SKW,WS}.
On the other hand, \Erase{Parrondo's paradox is introduced by Parrondo.}
the Parrondo paradox is the situation that a combined \Add{strategy} wins even if each \Add{strategy} loses.
The game with Parrondo's paradox is called Parrondo's game.
The Parrondo paradox has significant applications in many physical and biological systems like \cite{AD,PD}.
Additionally Parrondo's game via QW on the line has been investigated, such as \cite{CB,F,LZG,RB,MG}.
In the previous work, Parrondo's paradox is defined by the relationship between $P_{R}$ and $P_{L}$, where
$P_{R}$ is the probability of the quantum walker being found to the right of the origin and
$P_{L}$ is the probability of the quantum walker being found to the left of the origin.

Inspired by both models, we introduce a Parrondo's game based on QW search and propose 
Parrondo's paradox defined by the average of the success probability of finding marked vertices.
As far as we know, no previous study has studied the Parrondo game via QW search.
Furthermore we find the Parrondo paradox for our model \Add{by numerical simulations}, that is, bad search algorithms produce a good one \Erase{by numerical simulations}.
Besides we get rigorous results as well as numerical ones.
We prove the range in which the paradox occurs is symmetric with respect to the origin
on \Add{$T^{d}_{N}$} with even vertices and one marked vertex\Add{, where $T^{d}_{N}$ denotes the $d$-dimensional torus with $N^{d}$ vertices}.
In other words, our results show that bad search algorithms have the potential to become better algorithms by combining them.
To clarify the properties of the Parrondo game based on QW search will be a benefit for application to quantum information theory.

The rest of this paper is organized as follows.
In Section 2, we present the definition of Parrondo's game via QW search.
Section 3 deals with the numerical simulations for the Parrondo game on \Add{$T^{1}_{N}$ and $T^{2}_{N}$}.
In Section 4, we give a proof of our results on one marked \Add{$T^{d}_{N}$} with even vertices.
Section 5 concludes our results.

\clearpage

\section{{\bf Parrondo's game on QW search} \label{sec2}}
In this paper, we consider discrete-time QWs on $d$-regular graph with $N$ vertices.
The Hilbert space is given by $\mathcal{H}=\mathcal{H}^c\otimes \mathcal{H}^p$,
where $\mathcal{H}^c$ is the coin space spanned by the  orthonormal basis $\left\{ \ket{s}: s=0,1,...,d-1\right\}$ 
and $\mathcal{H}^p$ is the position space spanned by the orthonormal basis $\left\{ \ket{j}: j=0,1,...,N-1\right\}$.
The unitary operator described by $U=S\cdot C$ acts on $\mathcal{H}$, where $S$ is a shift operator and $C$ is a coin operator.
Under a search problem on a given graph, the coin operator is 
\begin{align}
  \label{eq:coin}
  C=C_{M}\otimes \sum_{v \in{M}}\ket{v}\bra{v}+C_{\bar{M}}\otimes \left(I_N-\sum_{v \in{M}}\ket{v}\bra{v}\right),
\end{align}
where $C_{M}$ and $C_{\bar{M}}$ are $d\times d$  matrices.
Here $M$ is the set of marked vertices whose number of elements is $m$.
This definition implies that $C_{M}$ operates marked vertices and $C_{\bar{M}}$ operates non-marked vertices. 

Parrondo's game based on the QW search algorithm is as follows:
We prepare two unitary operators $U_1$ and $U_2$ wrriten as 
\begin{align*}
  U_1=S\cdot C_1, \quad U_2=S\cdot C_2,
\end{align*}
where $C_1$ and $C_2$ have the same form given by Eq.\ (\ref{eq:coin}).
\Add{
  We consider $U_1$ and $U_2$ as strategies to find marked vertex, respectively.
}
In addition, a unitary operator $U_{(n_1,n_2)}$ combined $U_1$ and $U_2$ is denoted by
\begin{align*}
  U_{(n_1,n_2)}=(U_2)^{n_2}(U_1)^{n_1}
\end{align*}
for $n_1,n_2 \in \Z_{>}$, where $\Z_{>}$ is the set of positive integer.
\Add{
  We regard $U_{(n_1,n_2)}$ as a combined strategy of $U_1$ and $U_2$.
}
The initial state $\ket{\Psi_0}$ is the uniform state expressed as 
\begin{align*}
  \ket{\Psi_0}=\frac{1}{\sqrt{dN}}\sum^{d-1}_{s=0}\sum^{N-1}_{j=0}\ket{s,j}.
\end{align*}
\Add{
  Then we try to find marked vertex from all of vertices for a strategy $U$.
  If the marked vertex is found, we win and if not, we lose.
}
We define $\bar{p}_{(n_1,n_2)}(T)$, by
\begin{align}
  \label{eq:mean}
  \bar{p}_{(n_1,n_2)}(T)=\frac{1}{T} \sum \limits _{t=0}^{T-1} \sum \limits _{s=0}^{d-1} \sum \limits _{v \in M}|\bra{s,v}U_{(n_1,n_2)}^{t}\ket{\Psi_0}|^2.
\end{align}
Moreover, taking a limit as $T \to \infty$, we put
\begin{align}
  \label{eq:limit_mean}
  \bar{p}_{(n_1,n_2)}=\lim_{T \to \infty} \bar{p}_{(n_1,n_2)}(T)=\lim_{T \to \infty} \frac{1}{T} \sum \limits _{t=0}^{T-1} \sum \limits _{s=0}^{d-1} \sum \limits _{v \in M} |\bra{s,v}U_{(n_1,n_2)}^{t}\ket{\Psi_0}|^2,
\end{align}
if the right-hand side of Eq.\ (\ref{eq:limit_mean}) exists.
We should remark that $\bar{p}_{(1,0)}=\bar{p}_{(n,0)}$ and $\bar{p}_{(0,1)}=\bar{p}_{(0,n)}$ for $n \in \Z_{>}$.

\Add{
  It is noted that Equations (\ref{eq:mean}) and (\ref{eq:limit_mean}) are considered as the average probability of finding a marked vertex of a single graph for a fixed $T$ times and its limit with respect to $T$ (i.e., $T \rightarrow \infty$), respectively.
}

Here we introduce two types of the Parrondo paradox on QW search via Eq.\ (\ref{eq:limit_mean}).

\begin{definition}
  \mbox{} \\
  \label{def:parrondo}
  \rm{}
  \! If $\bar{p}_{(1,0)}<m/N$, \, $\bar{p}_{(0,1)}<m/N$, and \,$\bar{p}_{(n_1,n_2)}>m/N$ hold, we call it ``positive paradox''. \\
  \quad If $\bar{p}_{(1,0)}>m/N$, \, $\bar{p}_{(0,1)}>m/N$, and \,$\bar{p}_{(n_1,n_2)}<m/N$ hold, we call it ``negative paradox''.
\end{definition}
Note that in Definition \ref{def:parrondo}, the positive paradox means that the success probability is greater than $m/N$ for combined unitary operator $U$,
however, it is less than $m/N$ for $U_1$ only and $U_2$ only, respectively. 
By contrast, the negative paradox means that the success probability is less than $m/N$ for combined unitary operator $U$,
however, it is greater than $m/N$ for $U_1$ only and $U_2$ only, respectively.

\Add{
  In other words, the positive paradox means that a combination of losing strategies
  becomes a winning strategy on average.
  By contrast, the negative paradox means that a combination of winning strategies
  becomes a losing strategy on average.
}

\section{\bf {Numerical results} \label{sec3}}
From now on, we consider Parrondo's game on \Add{$T^{1}_{N}$ and $T^{2}_{N}$}.
We give the shift  operator $S$  and the coin operator $C(\alpha,\beta,\theta)$ \Add{of $T^{1}_{N}$}, by
\begin{align*}
  S= \ket{0}\bra{0} \otimes \sum \limits _{j=0}^{N-1} \ket{j+1}\bra{j}+\ket{1}\bra{1} \otimes \sum \limits _{j=0}^{N-1} \ket{j-1}\bra{j},
\end{align*}
\begin{align*}
  C(\alpha,\beta,\theta)=(-I_2) \otimes \sum_{v\in M}\ket{v}\bra{v}+ 
  \begin{pmatrix}
    e^{i\alpha}\cos{\theta} & e^{-i\beta}\sin{\theta} \\  
    e^{i\beta}\sin{\theta} & -e^{-i\alpha}\cos{\theta} \\
  \end{pmatrix} 
    \otimes \left(I_N-\sum_{v\in M}\ket{v}\bra{v}\right),
\end{align*}
where $\alpha,\beta,\theta \in [0,2\pi)$ and $I_2$ is the $2\times 2$ identity matrix.
\Add{Also unitary operator of $T^{2}_{N}$ is given by the tensor product of the unitary operators of $T^{1}_{N}$.}
We find both positive and negative paradoxes by changing parameters, see Figs.\ \ref{fig:1}, \ref{fig:2}, \ref{fig:3} and \Add{\ref{fig:4}}.
Moreover we simulate the range of $\theta \in [0,2\pi)$ in which the paradox occurs in Figs.\ \ref{fig:scatter1}, \ref{fig:scatter2}, \ref{fig:scatter3}, and \Add{\ref{fig:scatter4}}.
\Add{
  Note that Eq.\ (\ref{eq:limit_mean}) is used to define the paradox, however, we use Eq.\ (\ref{eq:mean}) in Section 3 because it is difficult to get the limit in simulation.
}

\begin{figure}[htbp]
  \begin{minipage}[b]{0.32\linewidth}
    \centering
    \includegraphics[keepaspectratio, scale=0.33]{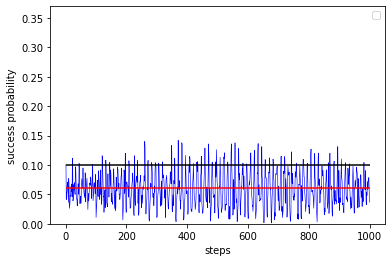}
    \subcaption{}
  \end{minipage}
  \begin{minipage}[b]{0.32\linewidth}
    \centering
    \includegraphics[keepaspectratio, scale=0.33]{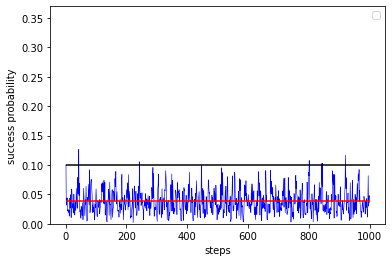}
    \subcaption{}
  \end{minipage} 
  \begin{minipage}[b]{0.32\linewidth}
    \centering
    \includegraphics[keepaspectratio, scale=0.33]{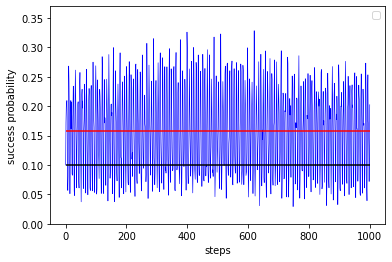}
    \subcaption{}
  \end{minipage} \\
  \begin{minipage}[b]{0.32\linewidth}
    \centering
    \includegraphics[keepaspectratio, scale=0.33]{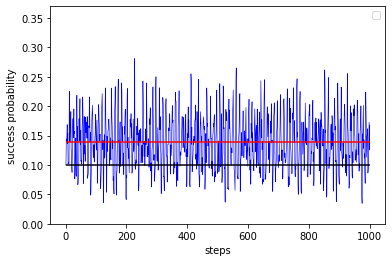}
    \subcaption{}
  \end{minipage}
  \begin{minipage}[b]{0.32\linewidth}
    \centering
    \includegraphics[keepaspectratio, scale=0.33]{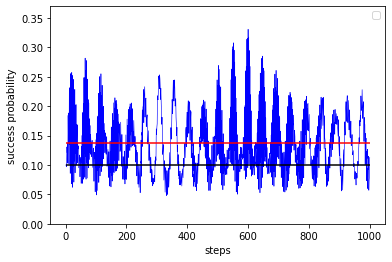}
    \subcaption{}
  \end{minipage} 
  \begin{minipage}[b]{0.32\linewidth}
    \centering
    \includegraphics[keepaspectratio, scale=0.33]{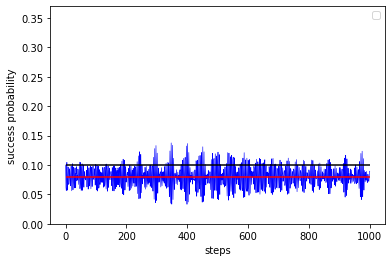}
    \subcaption{}
  \end{minipage}
  \caption{\Add{($T^{1}_{N}$ case)} The black and red lines correspond to $m/N$ and $\bar{p}_{(n_1,n_2)}(T)$, respectively and blue curve corresponds to the success probability.
          (c) combines (a) and (b). (f) combines (d) and (e).
          $N=10,\,T=1000,\,M=\{0\},\,m=1,\,(n_1,n_2)=(1,1)$.
          (a) $S\cdot C(4.90,0.06,3.53)$, (b) $S\cdot C(1.44,3.76,2.39)$,
          (d) $S\cdot C(1.69,3.84,3.64)$, (e) $S\cdot C(2.95,5.15,1.33)$.}
  \label{fig:1}
\end{figure}

\begin{figure}[htbp]
  \begin{minipage}[b]{0.32\linewidth}
    \centering
    \includegraphics[keepaspectratio, scale=0.33]{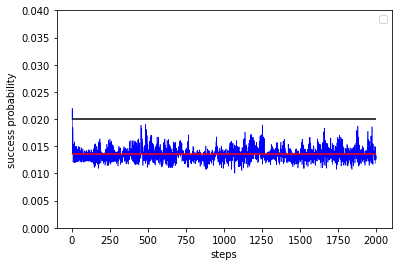}
    \subcaption{}
  \end{minipage}
  \begin{minipage}[b]{0.32\linewidth}
    \centering
    \includegraphics[keepaspectratio, scale=0.33]{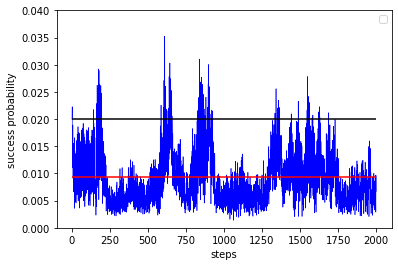}
    \subcaption{}
  \end{minipage} 
  \begin{minipage}[b]{0.32\linewidth}
    \centering
    \includegraphics[keepaspectratio, scale=0.33]{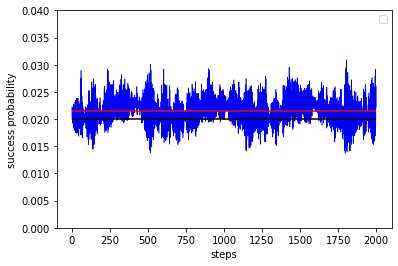}
    \subcaption{}
  \end{minipage} \\

  \begin{minipage}[b]{0.32\linewidth}
    \centering
    \includegraphics[keepaspectratio, scale=0.33]{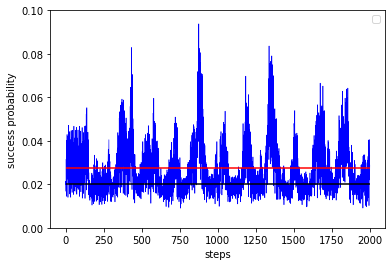}
    \subcaption{}
  \end{minipage}
  \begin{minipage}[b]{0.32\linewidth}
    \centering
    \includegraphics[keepaspectratio, scale=0.33]{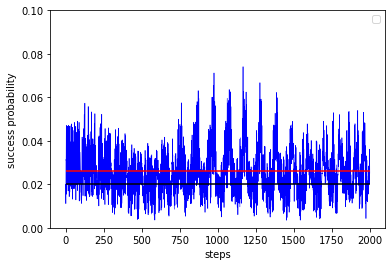}
    \subcaption{}
  \end{minipage} 
  \begin{minipage}[b]{0.32\linewidth}
    \centering
    \includegraphics[keepaspectratio, scale=0.33]{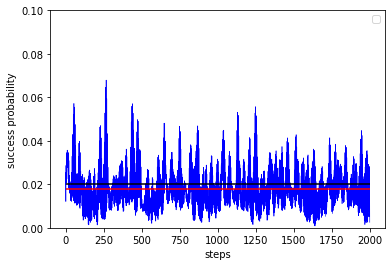}
    \subcaption{}
  \end{minipage} \\
  \caption{\Add{($T^{1}_{N}$ case)} The black and red lines correspond to $m/N$ and $\bar{p}_{(n_1,n_2)}(T)$, respectively and blue curve corresponds to the success probability.
  (c) combines (a) and (b). (f) combines (d) and (e).
          $N=50,\,T=2000,\,M=\{0\},\,m=1,\,(n_1,n_2)=(1,3)$.
            (a) $S\cdot C(4.85,4.93,1.91)$, (b) $S\cdot C(0.94,5.41,4.37)$,
            (d) $S\cdot C(5.78,5.41,1.93)$, (e) $S\cdot C(2.85,3.72,4.18)$.}
  \label{fig:2}
\end{figure}

\begin{figure}[htbp]
  \begin{minipage}[b]{0.32\linewidth}
    \centering
    \includegraphics[keepaspectratio, scale=0.33]{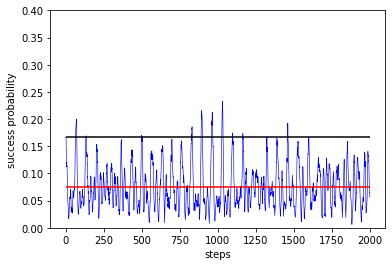}
    \subcaption{}
  \end{minipage}
  \begin{minipage}[b]{0.32\linewidth}
    \centering
    \includegraphics[keepaspectratio, scale=0.33]{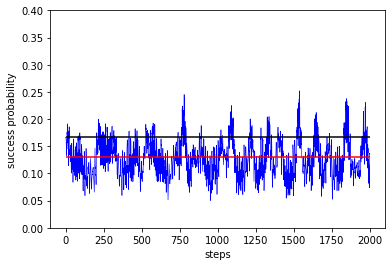}
    \subcaption{}
  \end{minipage} 
  \begin{minipage}[b]{0.32\linewidth}
    \centering
    \includegraphics[keepaspectratio, scale=0.33]{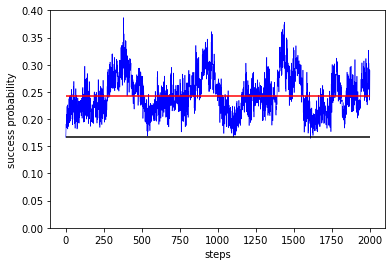}
    \subcaption{}
  \end{minipage} \\

  \begin{minipage}[b]{0.32\linewidth}
    \centering
    \includegraphics[keepaspectratio, scale=0.33]{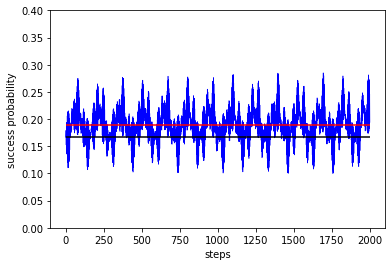}
    \subcaption{}
  \end{minipage}
  \begin{minipage}[b]{0.32\linewidth}
    \centering
    \includegraphics[keepaspectratio, scale=0.33]{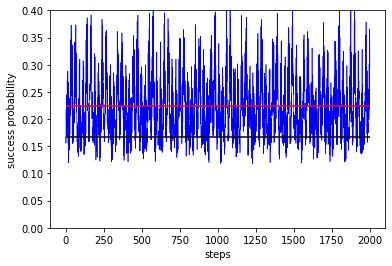}
    \subcaption{}
  \end{minipage} 
  \begin{minipage}[b]{0.32\linewidth}
    \centering
    \includegraphics[keepaspectratio, scale=0.33]{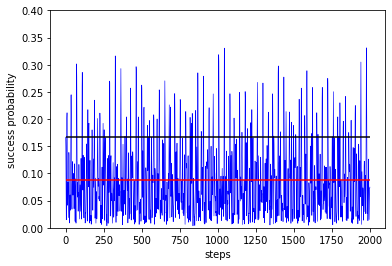}
    \subcaption{}
  \end{minipage} \\
  \caption{\Add{($T^{1}_{N}$ case)} The black and red lines correspond to $m/N$ and $\bar{p}_{(n_1,n_2)}(T)$, respectively and blue curve corresponds to the success probability.
  (c) combines (a) and (b). (f) combines (d) and (e).
            $N=30,\,T=2000,\,m=5,\,(n_1,n_2)=(1,1)$.
            (a),(d) $S\cdot C(2.58,2.60,2.05)$, (b),(e) $S\cdot C(3.49,4.83,1.85)$,
            (a)-(c) $M=\{0,1,2,3,4\}$, (d)-(f) $M=\{0,6,12,18,24\}$.}
  \label{fig:3}
\end{figure}

\begin{figure}[htbp]
  \begin{minipage}[b]{0.32\linewidth}
    \centering
    \includegraphics[keepaspectratio, scale=0.33]{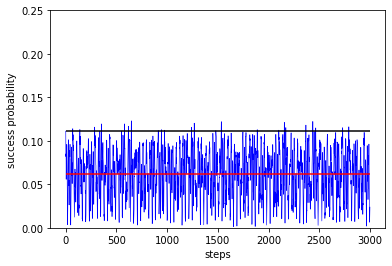}
    \subcaption{}
  \end{minipage}
  \begin{minipage}[b]{0.32\linewidth}
    \centering
    \includegraphics[keepaspectratio, scale=0.33]{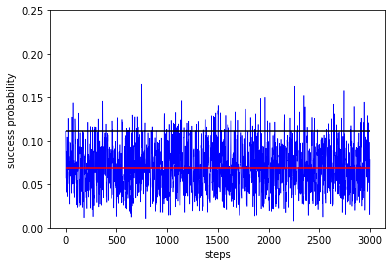}
    \subcaption{}
  \end{minipage} 
  \begin{minipage}[b]{0.32\linewidth}
    \centering
    \includegraphics[keepaspectratio, scale=0.33]{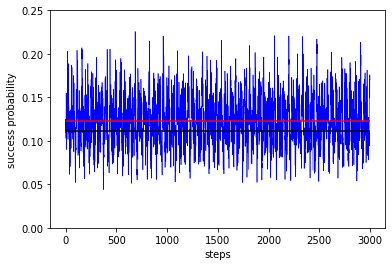}
    \subcaption{}
  \end{minipage} \\
  \begin{minipage}[b]{0.32\linewidth}
    \centering
    \includegraphics[keepaspectratio, scale=0.33]{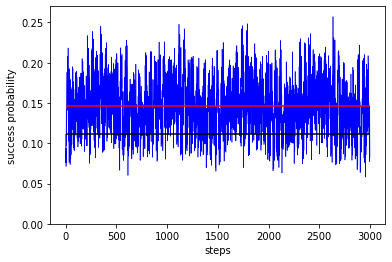}
    \subcaption{}
  \end{minipage}
  \begin{minipage}[b]{0.32\linewidth}
    \centering
    \includegraphics[keepaspectratio, scale=0.33]{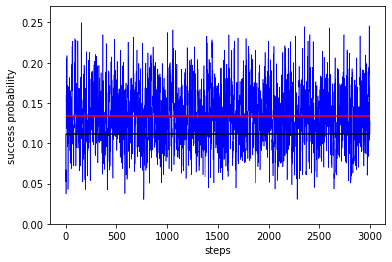}
    \subcaption{}
  \end{minipage} 
  \begin{minipage}[b]{0.32\linewidth}
    \centering
    \includegraphics[keepaspectratio, scale=0.33]{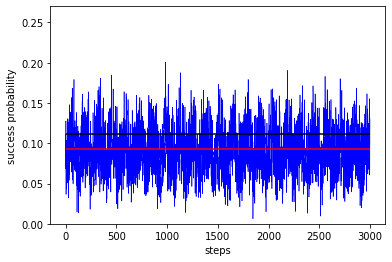}
    \subcaption{}
  \end{minipage}
  \caption{\Add{($T^{2}_{N}$ case)} The black and red lines correspond to $m/N$ and $\bar{p}_{(n_1,n_2)}(T)$, respectively and blue curve corresponds to the success probability.
          (c) combines (a) and (b). (f) combines (d) and (e).
          $N=9,\,T=3000,\,M=\{0\},\,m=1,\,(n_1,n_2)=(1,1)$.
          (a) $S\cdot C(1.38,0.95,0.23)$, (b) $S\cdot C(0.61,1.48,0.65)$,
          (d) $S\cdot C(3.10,3.08,0.42)$, (e) $S\cdot C(3.05,4.13,3.78)$.}
  \label{fig:4}
\end{figure}

\begin{figure}[htbp]
  \begin{minipage}[b]{0.25\linewidth}
    \centering
    \includegraphics[keepaspectratio, scale=0.25]{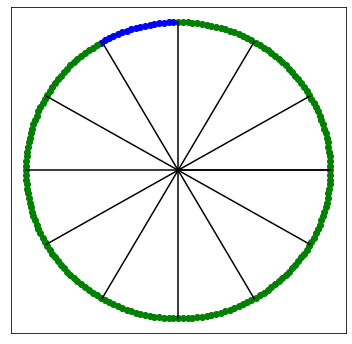}
    \subcaption{}
  \end{minipage}
  \begin{minipage}[b]{0.24\linewidth}
    \centering
    \includegraphics[keepaspectratio, scale=0.25]{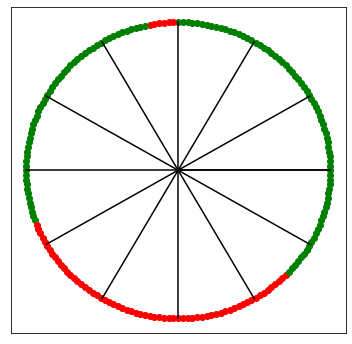}
    \subcaption{}
  \end{minipage} 
  \begin{minipage}[b]{0.24\linewidth}
    \centering
    \includegraphics[keepaspectratio, scale=0.25]{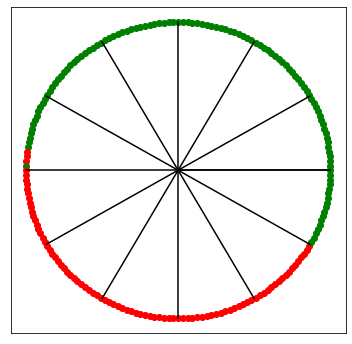}
    \subcaption{}
  \end{minipage}
  \begin{minipage}[b]{0.25\linewidth}
    \centering
    \includegraphics[keepaspectratio, scale=0.25]{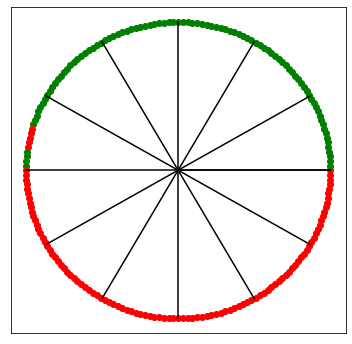}
    \subcaption{}
  \end{minipage} 
  \caption{\Add{($T^{1}_{N}$ case)} The range of $\theta \in [0,2\pi)$ in which the paradox occurs.
          The red and blue \Add{segment} correspond to positive paradox and negative paradox, respectively. The green  \Add{segment} corresponds that
          the paradox does not occur.
          $N=10,\,T=1000,\, (n_1,n_2)=(1,1)$,
          $U_1=S\cdot C(\pi/3,\pi/3,\pi/4),\,U_2=S\cdot C(\pi/3,\pi/3,\theta)$.
  (a) $M=\{0\},\, m=1$, (b) $M=\{0,1\},\, m=2$, (c) $M=\{0,1,2\},\, m=3$, (d) $M=\{0,1,2,3\},\, m=4$.}
  \label{fig:scatter1}
\end{figure}

\begin{figure}[htbp]
  \begin{minipage}[b]{0.25\linewidth}
    \centering
    \includegraphics[keepaspectratio, scale=0.25]{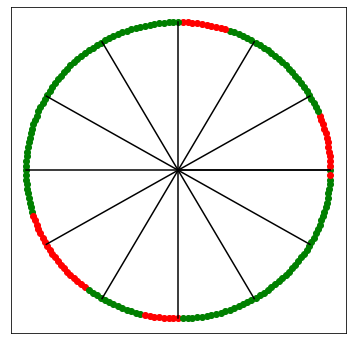}
    \subcaption{}
  \end{minipage}
  \begin{minipage}[b]{0.24\linewidth}
    \centering
    \includegraphics[keepaspectratio, scale=0.25]{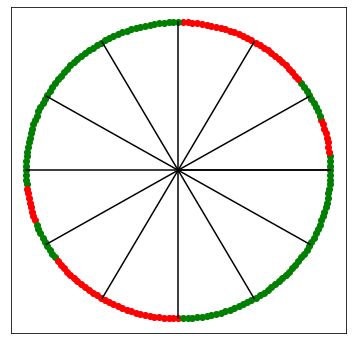}
    \subcaption{}
  \end{minipage} 
  \begin{minipage}[b]{0.24\linewidth}
    \centering
    \includegraphics[keepaspectratio, scale=0.25]{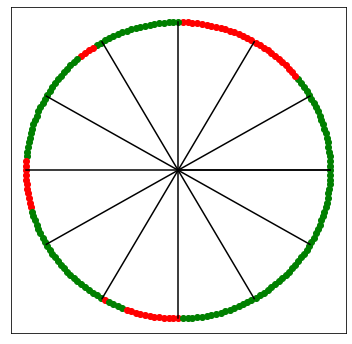}
    \subcaption{}
  \end{minipage}
  \begin{minipage}[b]{0.25\linewidth}
    \centering
    \includegraphics[keepaspectratio, scale=0.25]{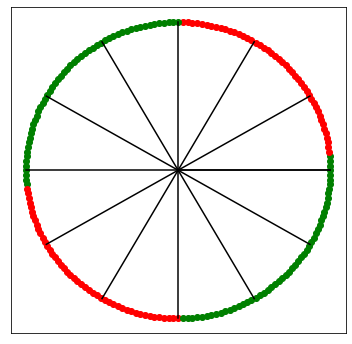}
    \subcaption{}
  \end{minipage} 
  \caption{\Add{($T^{1}_{N}$ case)} The range of $\theta \in [0,2\pi)$ in which the paradox occurs.
  The red  \Add{segment} corresponds to positive paradox. The green  \Add{segment} corresponds that
  the paradox does not occur.
    $T=1000,\,M=\{0\},\,m=1,\,(n_1,n_2)=(1,1),\,U_1=S\cdot C(4\pi/12,5\pi/12,23\pi/12),\quad U_2=S\cdot C(4\pi/12,5\pi/12,\theta)$.
  (a) $N=9$, (b) $N=10$, (c) $N=11$, (d) $N=12$.}
  \label{fig:scatter2}
\end{figure}

\begin{figure}[htbp]
  \begin{minipage}[b]{0.25\linewidth}
    \centering
    \includegraphics[keepaspectratio, scale=0.25]{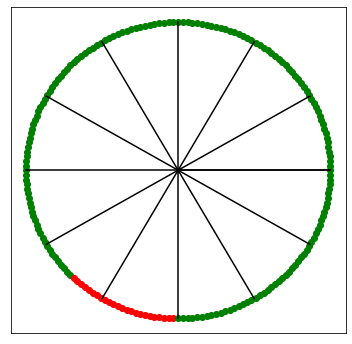}
    \subcaption{}
  \end{minipage}
  \begin{minipage}[b]{0.24\linewidth}
    \centering
    \includegraphics[keepaspectratio, scale=0.25]{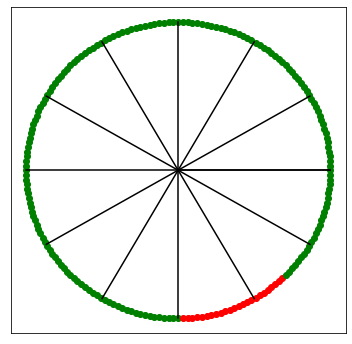}
    \subcaption{}
  \end{minipage} 
  \begin{minipage}[b]{0.24\linewidth}
    \centering
    \includegraphics[keepaspectratio, scale=0.25]{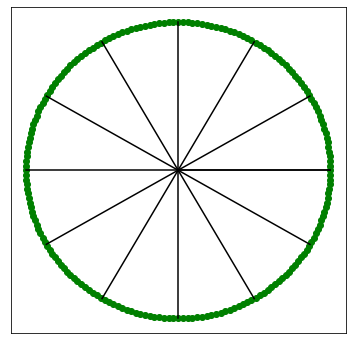}
    \subcaption{}
  \end{minipage}
  \begin{minipage}[b]{0.25\linewidth}
    \centering
    \includegraphics[keepaspectratio, scale=0.25]{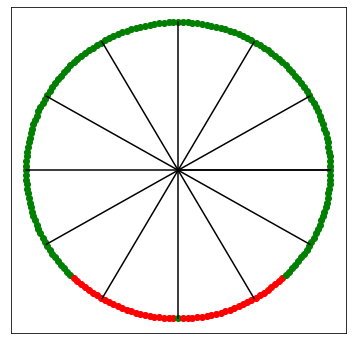}
    \subcaption{}
  \end{minipage} 
  \caption{\Add{($T^{1}_{N}$ case)} The range of $\theta \in [0,2\pi)$ in which the paradox occurs.
  The red  \Add{segment} corresponds to positive paradox. The green  \Add{segment} corresponds that
  the paradox does not occur.
    $N=10,\,M=\{0\},\,m=1,\,T=1000$,
  $U_1=S\cdot C(\pi/3,2\pi/3,\pi/4),\, U_2=S\cdot C(\pi/4,\pi/4,\theta)$.
  (a) $(n_1,n_2)=(1,1)$, (b) $(n_1,n_2)=(1,2)$, (c) $(n_1,n_2)=(2,1)$, (d) $(n_1,n_2)=(2,2)$.}
  \label{fig:scatter3}
\end{figure}

\begin{figure}[htbp]
  \begin{minipage}[b]{0.25\linewidth}
    \centering
    \includegraphics[keepaspectratio, scale=0.25]{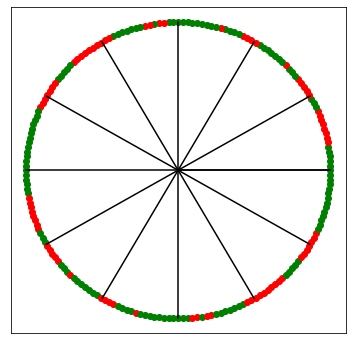}
    \subcaption{}
  \end{minipage}
  \begin{minipage}[b]{0.24\linewidth}
    \centering
    \includegraphics[keepaspectratio, scale=0.25]{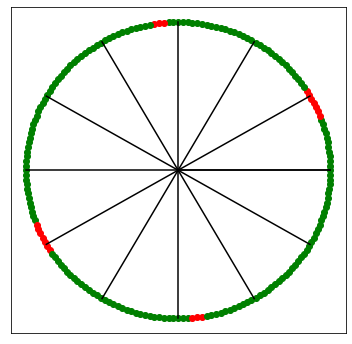}
    \subcaption{}
  \end{minipage} 
  \begin{minipage}[b]{0.24\linewidth}
    \centering
    \includegraphics[keepaspectratio, scale=0.25]{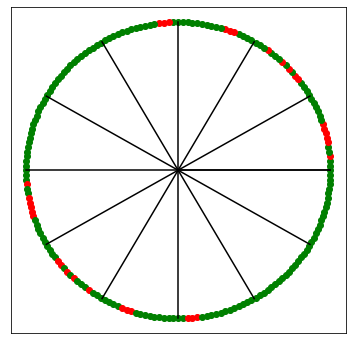}
    \subcaption{}
  \end{minipage}
  \begin{minipage}[b]{0.25\linewidth}
    \centering
    \includegraphics[keepaspectratio, scale=0.25]{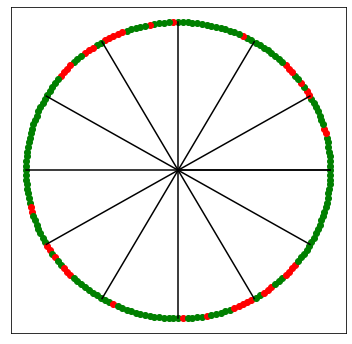}
    \subcaption{}
  \end{minipage} 
  \caption{\Add{($T^{2}_{N}$ case)} The range of $\theta \in [0,2\pi)$ in which the paradox occurs.
  The red  \Add{segment} corresponds to positive paradox. The green  \Add{segment} corresponds that
  the paradox does not occur.
    $T=1000,\,M=\{0\},\,m=1,\,(n_1,n_2)=(1,1),\,U_1=S\cdot C(4\pi/12,5\pi/12,23\pi/12),\quad U_2=S\cdot C(4\pi/12,5\pi/12,\theta)$.
  (a) $N=9$, (b) $N=16$, (c) $N=25$, (d) $N=36$.}
  \label{fig:scatter4}
\end{figure}

\clearpage
\section{{\bf Rigorous results} \label{sec4}}
This section deals with some rigorous results, inspierd by numerical results in Section 3, in particular, Figs.\ \ref{fig:scatter2} and \Add{
  \ref{fig:scatter4}}.
\Add{We rewrite coin space $\left\{ \ket{0}, \ket{1}\right\}$ as
\begin{align*}
  \mathcal{H}^{c}=\{\ket{L} , \ket{R}\}
\end{align*}
to clearly distinguish $\ket{0}$ in coin space from $\ket{0}$ in position space, where $\ket{L}=\ket{0}$ and $\ket{R}=\ket{1}$.
}
We suppose that the coin operators are given by 
\begin{align*}
  C_1&= (-I_2)\otimes\ket{0}\bra{0}+ C'_1\otimes(I_N-\ket{0}\bra{0}), \\
  C_2&= (-I_2)\otimes\ket{0}\bra{0}+ C'_2\otimes(I_N-\ket{0}\bra{0}), \\
  \widetilde{C_2}&=(-I_2)\otimes\ket{0}\bra{0}-C'_2\otimes(I_N-\ket{0}\bra{0}), \\
  U_1&=S\cdot C_1 , \qquad U_2=S\cdot C_2, \qquad \widetilde{U_2}=S\cdot \widetilde{C_2}. \\
\end{align*}
Here $C'_1$ and $C'_2$ are $2 \times 2$ matrices.
Then we show the following result.

\begin{theorem}
  \label{thm:d}
  For $T^{d}_{N}$ with $N=$ even, we have
  \begin{align*}
    \sum_{s \in \{L,R\}}|(\bra{s,0})^{\otimes d}\{(U_2U_1)^{\otimes d}\}^{t}\ket{\Psi}|^2=\sum_{s \in \{L,R\}}|(\bra{s,0})^{\otimes d}\{(\widetilde{U_2}U_1)^{\otimes d}\}^{t}\ket{\Psi}|^2,
  \end{align*}
  for $t \in \Z_{>}$ and $\Psi \in \mathcal{H}$.
\end{theorem}

\Add{Before giving the proof of Theorem \ref{thm:d} ($T^{d}_{N}$ case), for a better understanding, we will present the corresponding proof for $T^{1}_{N}$ case (i.e., $d=1$).}

\begin{theorem}
  \label{thm:1}
  For $T^{1}_N$ with even vertices, we have 
  \begin{align*}
    \sum_{s \in \{L,R\}}|\bra{s,0}(U_2U_1)^{t}\ket{\Psi}|^2=\sum_{s \in \{L,R\}}|\bra{s,0}(\widetilde{U_2}U_1)^{t}\ket{\Psi}|^2,
  \end{align*}
  for $t \in \Z_{>}$ and $\Psi \in \mathcal{H}$.
\end{theorem}

\begin{proof}[{\bf Proof of Theorem \ref{thm:1}}]
  We will \Add{solve} the following equation by induction with respected to $t$:
  \begin{align}
    \label{eq:mat}
    (U_2U_1)^t=
    X\left((U_2U_1)^{t}\right)
    \otimes 
    \Add{Y},
  \end{align}
  \Add{
    where $X(A)$ is a $2 \times 2$ matrix determined by a $2N \times 2N$ matrix $A$
    and $Y=\left(y(i,j)\right)_{i,j=1,2,...,N}$ is an $N \times N$ matrix where $(i,j)$ component, i.e,
    $y(i,j)$, with $i+j=$ odd is zero and that with $i+j=$ even is an arbitrary complex number denoted by $*$.
  }
  \Add{
    For example, if $N=4$, then $Y$ is given by
    \begin{align*}
      Y=
      \begin{pmatrix}
        * & 0 & * & 0 \\
        0 & * & 0 & * \\
        * & 0 & * & 0 \\
        0 & * & 0 & * \\
      \end{pmatrix}.
    \end{align*}
  }
  When $t=1$, we get
  \begin{align}
    \label{eq:mat_t=1}
    U_1
    &=S\cdot C_1 \notag\\
    &=\{\Add {\ket{L}\bra{L}} \otimes \sum \limits _{j=0}^{N-1} \ket{j+1}\bra{j}+\Add {\ket{R}\bra{R}} \otimes \sum \limits _{j=0}^{N-1} \ket{j-1}\bra{j}\} \notag\\
    & \qquad \cdot \{(-I_2)\otimes\ket{0}\bra{0}+ C'_1\otimes(I_N-\ket{0}\bra{0})\} \notag\\
    &=-\Add {\ket{L}\bra{L}} \otimes \ket{1}\bra{0}-\Add {\ket{R}\bra{R}} \otimes \ket{N-1}\bra{0} \notag\\
    & \qquad +\Add {\ket{L}\bra{L}} C'_1 \otimes \sum \limits _{j=1}^{N-1} \ket{j+1}\bra{j}+ \Add {\ket{R}\bra{R}} C'_1 \otimes \sum \limits _{j=1}^{N-1} \ket{j-1}\bra{j}. 
  \end{align}
  Then Eq.\ (\ref{eq:mat_t=1}) is rewritten as 
  \begin{align}
    \label{eq:mat_1}
    U_1=
    X(U_1) \otimes
    \begin{pmatrix}
      0 & * & \dots  & 0 & * \\
      * & 0 & \dots  & 0 & 0 \\
      \vdots & \vdots & \ddots & \vdots &\vdots \\
      0 & 0 & \dots  & 0 & * \\
      * & 0 & \dots & * & 0
    \end{pmatrix}.
  \end{align}
  By using a similar method, $U_2$ has the same form of the right-hand side of Eq.\ (\ref{eq:mat_1}).
  Therefore we see
  \begin{align*}
    U_2U_1
    &= \left\{X(U_2) \otimes
    \begin{pmatrix}
      0 & * & \dots  & 0 & * \\
      * & 0 & \dots  & 0 & 0 \\
      \vdots & \vdots & \ddots & \vdots &\vdots \\
      0 & 0 & \dots  & 0 & * \\
      * & 0 & \dots & * & 0
    \end{pmatrix}
    \right\}
    \cdot 
    \left\{
      X(U_1) \otimes
    \begin{pmatrix}
      0 & * & \dots  & 0 & * \\
      * & 0 & \dots  & 0 & 0 \\
      \vdots & \vdots & \ddots & \vdots &\vdots \\
      0 & 0 & \dots  & 0 & * \\
      * & 0 & \dots & * & 0
    \end{pmatrix} 
    \right\} \\
    &=X(U_2U_1)\otimes
    \begin{pmatrix}
      * & 0 & \dots  & * & 0 \\
      0 & * & \dots  & 0 & * \\
      \vdots & \vdots & \ddots & \vdots &\vdots \\
      * & 0 & \dots  & * & 0 \\
      0 & * & \dots & 0 & *
    \end{pmatrix}.
  \end{align*}
  \Add{We should remark $X(U_2)X(U_1)=X(U_2U_1)$.}
  Hence Eq.\ (\ref{eq:mat}) is correct for $t=1$.
  Assume that Eq.\ (\ref{eq:mat}) holds for $t$.
  When $t+1$, we compute
  \begin{align*}
    (U_2U_1)^{t+1}
    &=(U_2U_1)^{t} \cdot U_2U_1 \\
    &= \left\{X\left((U_2U_1)^{t}\right) \otimes 
    \begin{pmatrix}
      * & 0 & \dots  & * & 0 \\
      0 & * & \dots  & 0 & * \\
      \vdots & \vdots & \ddots & \vdots &\vdots \\
      * & 0 & \dots  & * & 0 \\
      0 & * & \dots & 0 & *
    \end{pmatrix}\right\}
    \cdot 
    \left\{X(U_2U_1) \otimes 
    \begin{pmatrix}
      * & 0 & \dots  & * & 0 \\
      0 & * & \dots  & 0 & * \\
      \vdots & \vdots & \ddots & \vdots &\vdots \\
      * & 0 & \dots  & * & 0 \\
      0 & * & \dots & 0 & *
    \end{pmatrix}\right\} \\
    &=X\left((U_2U_1)^{t+1}\right) \otimes 
    \begin{pmatrix}
      * & 0 & \dots  & * & 0 \\
      0 & * & \dots  & 0 & * \\
      \vdots & \vdots & \ddots & \vdots &\vdots \\
      * & 0 & \dots  & * & 0 \\
      0 & * & \dots & 0 & *
    \end{pmatrix}.
  \end{align*}
  By induction, Eq.\ (\ref{eq:mat}) holds for any $t \in \mathbb{Z}_{>}$.
  Then it follows from Eq.\ (\ref{eq:mat}) that
  \begin{align}
    \label{eq:0}
    \bra{s,0}(U_2U_1)^{t} \ket{\Add{L},1}=\bra{s,0}(U_2U_1)^{t}\ket{\Add{R},N-1}=0.
  \end{align}
  On the other hand, the definition of the coin operator yields
  \begin{align*}
    C_2+\widetilde{C_2}=-2I_2 \otimes \ket{0}\bra{0}.
  \end{align*}
  Thus we obtain
  \begin{align}
    \label{eq:sum}
    U_2+\widetilde{U_2}
    &=S\cdot (C_2+\widetilde{C_2}) \notag\\
    &=\big\{\Add {\ket{L}\bra{L}} \otimes \sum \limits _{j=0}^{N-1} \ket{j+1}\bra{j}+\Add {\ket{R}\bra{R}} \otimes \sum \limits _{j=0}^{N-1} \ket{j-1}\bra{j} \big\} \cdot \left(-2I_2 \otimes \ket{0}\bra{0} \right) \notag\\
    &=-2\left( \Add {\ket{L}\bra{L}} \otimes \ket{1}\bra{0} + \Add {\ket{R}\bra{R}} \otimes \ket{N-1}\bra{0}\right) \notag\\
    &=-2\ket{\Add{L},1}\bra{\Add{L},0}-2\ket{\Add{R},N-1}\bra{\Add{R},0}.
  \end{align}
  We will show the following equation by induction with respected to $t$:
  \begin{align}
    \label{eq:induction}
    \bra{s,0}(\widetilde{U_2}U_1)^t\ket{\Psi}=(-1)^t\bra{s,0}(U_2U_1)^t\ket{\Psi}.
  \end{align}
  When $t=1$, Eq.\ (\ref{eq:sum}) gives
  \begin{align}
    \label{eq:u2u1}
    \bra{s,0}\widetilde{U_2}U_1\ket{\Psi}
    &=-\bra{s,0}U_2U_1\ket{\Psi}-2\braket{s,0|\Add{L},1}\bra{\Add{L},0}U_1\ket{\Psi}
     -2\braket{s,0|\Add{R},N-1}\bra{\Add{R},0}U_1\ket{\Psi} \notag\\
    &=-\bra{s,0}U_2U_1\ket{\Psi}.
  \end{align}
  Hence Eq.\ (\ref{eq:induction}) is correct for $t=1$.
  Next we assume that Eq.\ (\ref{eq:induction}) holds for $t \Erase{(\geq 1)}$.
  When $t+1$, the assumption on $t$ implies
  \begin{align*}
    \bra{s,0}(\widetilde{U_2}U_1)^{t+1}\ket{\Psi}
    &=\bra{s,0} (\widetilde{U_2}U_1)^{t}\widetilde{U_2}U_1\ket{\Psi} \\
    &=(-1)^{t}\bra{s,0} (U_2U_1)^{t}\widetilde{U_2}U_1\ket{\Psi} \\
    &=(-1)^{t}\bra{s,0}(U_2U_1)^{t} \Add{ \Big\{ }-U_2-2\ket{\Add{L},1}\bra{\Add{L},0}
    -2\ket{\Add{R},N-1}\bra{\Add{R},0} \Add {\Big\} }U_1\ket{\Psi}.
  \end{align*}
  Therefore we see
  \begin{align}
    \label{eq:t+1}
    \bra{s,0}(\widetilde{U_2}U_1)^{t+1}\ket{\Psi} 
    &=(-1)^{t+1}\bra{s,0}(U_2U_1)^{t+1}\ket{\Psi} \notag\\
    &\quad -2\cdot (-1)^{t}\bra{s,0}(U_2U_1)^{t}\ket{\Add{L},1}\bra{\Add{L},0}U_1\ket{\Psi} \notag\\
    &\quad -2\cdot (-1)^{t}\bra{s,0}(U_2U_1)^{t}\ket{\Add{R},N-1}\bra{\Add{R},0}U_1\ket{\Psi}.
  \end{align}
  Combining Eq.\ (\ref{eq:0}) with Eq.\ (\ref{eq:t+1}) yields
  \begin{align*}
    \bra{s,0}(\widetilde{U_2}U_1)^{t+1}\ket{\Psi}
    =(-1)^{t+1}\bra{s,0}(U_2U_1)^{t+1}\ket{\Psi}.
  \end{align*}
  By induction, 
  Eq.\ (\ref{eq:induction}) is true for any $t \in \Z_{>}$.
  Then it follows from Eq.\ (\ref{eq:induction}) that 
  we have the desired conclusion:
  \begin{align*}
    \sum_{\Add{s \in \{L,R\}}}|\bra{s,0}(\widetilde{U_2}U_1)^{t}\ket{\Psi}|^2
    &=\sum_{\Add{s \in \{L,R\}}}|(-1)^n\bra{s,0}(U_2U_1)^t\ket{\Psi}|^2 \\
    &=\sum_{\Add{s \in \{L,R\}}}|\bra{s,0}(U_2U_1)^t\ket{\Psi}|^2.
  \end{align*}
\end{proof}

\Add{From now on, in a similar way, we consider the corresponding result for the general $T^{d}_{N}$ case.}


  \begin{proof}[\Add{{\bf Proof of Theorem \ref{thm:d}}}]
    We will show the following equation by induction with respected to $t$:
  \begin{align}
    \label{eq:induction_d}
    (\bra{s,0})^{\otimes d}\{(\widetilde{U_2}U_1)^{\otimes d}\}^t\ket{\Psi}=(-1)^{dt}(\bra{s,0})^{\otimes d}\{(U_2U_1)^{\otimes d}\}^t\ket{\Psi}.
  \end{align}
  When $t=1$, by using notations in Eqs.\ (\ref{eq:sum}) and (\ref{eq:u2u1}), we get
  \begin{align*}
    (\bra{s,0})^{\otimes d} (\widetilde{U_2}U_1)^{\otimes d}\ket{\Psi}
    &=(\bra{s,0}\widetilde{U_2}U_1)^{\otimes d} \ket{\Psi} \\
    &=(-\bra{s,0}U_2U_1)^{\otimes d}\ket{\Psi} \\
    &=(-1)^{d}(\bra{s,0})^{\otimes d} (U_2U_1)^{\otimes d} \ket{\Psi}.
  \end{align*}
  Hence Eq.\ (\ref{eq:induction_d}) is correct for $t=1$.
  Next we assume that Eq.\ (\ref{eq:induction_d}) holds for $t$.
  When $t+1$, we compute 
  \begin{align*}
    (\bra{s,0})^{\otimes d}\{(\widetilde{U_2}U_1)^{\otimes d}\}^{t+1}\ket{\Psi}
    &=(\bra{s,0})^{\otimes d} \{(\widetilde{U_2}U_1)^{\otimes d}\}^{t} (\widetilde{U_2}U_1)^{\otimes d}\ket{\Psi} \\
    &=(-1)^{dt}(\bra{s,0})^{\otimes d} \{(U_2U_1)^{\otimes d}\}^{t}(\widetilde{U_2}U_1)^{\otimes d}\ket{\Psi} \\
    &=(-1)^{dt}(\bra{s,0})^{\otimes d} \{(U_2U_1)^{\otimes d}\}^{t} \\
    &\quad \times ( -U_2U_1-2\ket{L,1}\bra{L,0}U_1-2\ket{R,N-1}\bra{R,0}U_1 )^{\otimes d}\ket{\Psi}.
  \end{align*}
  Therefore we see
  \begin{align}
    \label{eq:t+1_d}
    (\bra{s,0})^{\otimes d}\{(\widetilde{U_2}U_1)^{\otimes d}\}^{t+1}\ket{\Psi}
    &=(-1)^{dt} \Big\{-\bra{s,0}(U_2U_1)^{t+1} -2\bra{s,0}(U_2U_1)^{t}\ket{L,1}\bra{L,0}U_1 \notag\\
    &\quad -2\bra{s,0}(U_2U_1)^{t}\ket{R,N-1}\bra{R,0}U_1 \Big\}^{\otimes d}\ket{\Psi}.
  \end{align}
  Combining Eq.\ (\ref{eq:0}) with Eq.\ (\ref{eq:t+1_d}) yields
  \begin{align*}
    (\bra{s,0})^{\otimes d}\{(\widetilde{U_2}U_1)^{\otimes d}\}^{t+1}\ket{\Psi}
    &=(-1)^{dt}\{-\bra{s,0}(U_2U_1)^{t+1}\}^{\otimes d} \ket{\Psi} \\
    &=(-1)^{d(t+1)}(\bra{s,0})^{\otimes d}\{(U_2U_1)^{\otimes d}\}^{t+1}\ket{\Psi}.
  \end{align*}
  By induction, Eq.\ (\ref{eq:induction_d}) is true for any $t \in \Z_{>}$.
  Then it follows from Eq.\ (\ref{eq:induction_d}) that 
  we have the desired conclusion:
  \begin{align*}
    \sum_{s \in \{L,R\}}|(\bra{s,0})^{\otimes d}\{(U_2U_1)^{\otimes d}\}^{t}\ket{\Psi}|^2
    &=\sum_{s \in \{L,R\}}|(-1)^{dt}(\bra{s,0})^{\otimes d}\{(U_2U_1)^{\otimes d}\}^t\ket{\Psi}|^2 \\
    &=\sum_{s \in \{L,R\}}|(\bra{s,0})^{\otimes d}\{(U_2U_1)^{\otimes d}\}^t\ket{\Psi}|^2.
  \end{align*}
  \end{proof}

  We should remark that Theorem \ref{thm:d} implies that the success probability for $U_2U_1$ is 
equal to that for $\widetilde{U_2}U_1$.
Hence Theorem \ref{thm:d} gives a proof that the range in which the paradox occurs
is symmetric across the origin, see Figs.\ \ref{fig:scatter2} \Add{and \ref{fig:scatter4}}.

\section{Conclusion \label{sec5}}
The present paper proposed a new type of Parrondo's game via QW search and we discovered both positive and negative paradoxes on \Add{$T^{1}_{N}$ and $T^{2}_{N}$}.
In addition, our numerical simulations confirmed that the paradox exists for some parameters,
such as the number of vertices $N$, the number of marked vertices $m$ and the number of the combination of two unitary operators $(n_1,n_2)$.
Moreover we show the range in which the paradox occurs is symmetric about the origin on \Add{$T^{d}_{N}$} with even vertices and one marked vertex.
One of the future problems would be to investigate how paradoxes behave on other graphs, for example, complete graph and hypercube graph\Erase{and $d$-dimensional torus}.
Another interesting problem is to analyze the parameters of the coin operators generating paradoxes.
We think that the Parrondo game on QW search has one of the possibilities to improve quantum search \Add{by combining} bad search algorithms.
\section*{Data Availability}
Our manuscript has no associated data.

\section*{Conflicts of interest}
The authors declare no conflict of interest.

\end{document}